\documentclass[conference]{ieeeconf}

\IEEEoverridecommandlockouts  
\usepackage[utf8]{inputenc} % allow utf-8 input
\usepackage[T1]{fontenc}    % use 8-bit T1 fonts
\usepackage{hyperref}       % hyperlinks
\usepackage{url}            % simple URL typesetting
\usepackage{booktabs}   
\usepackage{amsmath}% professional-quality tables
\usepackage{amsfonts}
\usepackage{xcolor}
\usepackage{soul}
\usepackage{bm}
\usepackage{nicefrac}       % compact symbols for 1/2, etc.
\usepackage{microtype}      % microtypography
\usepackage{lipsum}
\usepackage{fancyhdr}       % header
\usepackage{graphicx}       % graphics
\graphicspath{{media/}}     % organize your images and other figures under media/ folder
\usepackage{caption}
\usepackage{subcaption}
\usepackage{stmaryrd}
\usepackage{pgfplots}
\pgfplotsset{compat=1.18}
\usepackage{algorithm} 
\usepackage{algpseudocode}

\usepackage[utf8]{inputenc} % allow utf-8 input
\usepackage[T1]{fontenc}    % use 8-bit T1 fonts
\usepackage{hyperref}       % hyperlinks
\usepackage{url}            % simple URL typesetting
\usepackage{booktabs}   
\usepackage{amsmath}% professional-quality tables
\usepackage{amsthm}
\usepackage{amsfonts}
\usepackage{mdframed}
\usepackage{graphicx}
\usepackage{xcolor}
\usepackage{epstopdf}
\usepackage{amssymb}
\usepackage{matlab-prettifier}
\usepackage{siunitx}
\usepackage{soul}
\usepackage{nicefrac}       % compact symbols for 1/2, etc.
\usepackage{microtype}      % microtypography
\usepackage{lipsum}
\usepackage{fancyhdr}  
\usepackage[ragged]{sidecap}    

\usepackage{graphicx}       % graphics
\graphicspath{{media/}}     % organize your images and other figures under media/ folder
\usepackage{caption}
\usepackage{subcaption}
\usepackage{stmaryrd}
\usepackage{algorithm} 
\usepackage{algpseudocode} \usepackage{cuted}

\newtheoremstyle{thmv2}% name
{}%      Space above, empty = `usual value'
{.1em}%      Space below
{\itshape}% Body font
{}%         Indent amount (empty = no indent, \parindent = para indent)
{\bfseries}% Thm head font
{.}%        Punctuation after thm head
{.1em}% Space after thm head: \newline = linebreak
{}%         Thm head spec
\theoremstyle{thmv2}

\usepackage{mathtools}
\usepackage{tikz}
\usetikzlibrary{positioning,arrows.meta}
\usetikzlibrary{shapes,arrows,automata}
\usetikzlibrary{intersections}
\usetikzlibrary{calc,3d,fit,backgrounds,matrix}
\usepackage{cancel}
\usepackage[mathscr]{eucal}
\usepackage{algpseudocode}

\newcommand{\setwitness}{\mathcal{W}}
\newcommand{\setagent}{\mathcal{I}}
\newcommand{\eventually}{\Diamond}
\newcommand{\until}{\mathbin{\sf U}}

\newcommand{\spaceX}{\mathbb{X}} %state space

\newcommand{\mbP}{\mathbb{P}}

\newcommand{\mdpM}{\mathbf{M}}
\newcommand{\Tr}{{\mathbb T}}
\newcommand{\AP}{\mathrm{AP}}
\newcommand{\letter}{l}
\newcommand{\word}{\boldsymbol{l}}

\newcommand{\DFA}{\mathcal A}

\newcommand{\spaceA}{\mathbb{A}} %action space

\newcommand{\valuemapping}{v}

\newcommand{\tensorV}{\mathscr{V}}
\newcommand{\qmapping}{\mathcal{L}_Q}

\newcommand{\expectation}{\mathbb{E}}

\newcommand{\op}{\mathbf{T}}

\newtheorem{definition}{Definition}
\newtheorem{myremark}{Remark}
\newenvironment{tightdisplay}{%
  \begingroup
  \setlength\abovedisplayskip{6pt}%
  \setlength\belowdisplayskip{6pt}%
  \setlength\abovedisplayshortskip{4pt}%
  \setlength\belowdisplayshortskip{4pt}%
}{\endgroup}

\newtheorem{example}{Example}

\newtheorem{lemma}{Lemma}
\newtheorem{theorem}{Theorem}

\newtheorem{proposition}{Proposition}

\begin{document}
\title{
%Tensor-Based Dual-Tree Dynamic Programming for Homogeneous Multi-Agent Synthesis under Counting LTL\\
Compressing Correct-by-Design Synthesis for Stochastic Homogeneous Multi-Agent Systems with Counting LTL
}
\author{Xinyuan Qiu$^{\dagger}$, Ruohan Wang$^{\dagger}$, Siyuan Liu, Sofie Haesaert \thanks{$^{\dagger}$ The authors contributed to this paper equally.}\thanks{This work is supported by the European project SymAware under the grant number 101070802 and COVER under the grant number 101086228. The authors are with the Department of Electrical Engineering, Control
Systems Group, Eindhoven University of Technology, The Netherlands.  Emails: x.qiu@student.tue.nl; $\{$r.wang2,s.liu5,s.haesaert$\}$@tue.nl.}
}

\maketitle
\begin{abstract}
Correct-by-design synthesis provides a principled framework for establishing formal safety guarantees for stochastic multi-agent systems (MAS). However, conventional approaches based on finite abstractions often incur prohibitive computational costs as the number of agents and the complexity of temporal logic specifications increase. In this work, we study homogeneous stochastic MAS under counting linear temporal logic (cLTL) specifications, and show that the corresponding satisfaction probability admits a structured tensor decomposition via leveraging deterministic finite automata (DFA). Building on this structure, we develop a dual-tree-based value iteration framework that reduces redundant computation in the process of dynamic programming. Numerical results demonstrate the proposed approach's effectiveness and scalability for complex specifications and large-scale MAS.
%Recent work has demonstrated the effectiveness of low rank value iteration for stochastic systems under syntactically co-safe LTL (scLTL) specifications, while remaining limitations on handling the counting constraints that arise in multi-agent coordination. As the complexity of specifications grows, the redundant computation of that approach becomes increasingly severe, leading to poor scalability in multi-agent systems (MAS). %In this work, we proposed a Dual Tree-based value iteration framework exploiting the decomposition structure from the satisfaction probability of decoupled homogeneous stochastic MAS under counting LTL (cLTL) specifications.
\end{abstract}

% \begin{keywords}
% no need to include keywords in CDC
% \end{keywords}

\section{Introduction}
\label{sec:introduction}
The applications of multi-agent systems (MAS) in modern societies have expanded rapidly, ranging from entertainment-oriented drone formations to macroeconomics systems involving multiple decision-making entities, such as traffic management~\cite{bauza2013traffic} and wilderness search and rescue~\cite{macwan2014multirobot}. These applications require the coordination of increasingly large collections of agents, ensuring system safety and developing scalable synthesis techniques have become increasingly important. For safety-critical systems, formal verification approaches~\cite{baier2008principles,belta2017formal} offer mathematically rigorous tools for proving that a system satisfies prescribed safety properties. More recently, temporal logics have emerged as a powerful framework for formally specifying complex tasks and safety requirements, as they enable the rigorous and unambiguous description of propositions and their evolution over time. Building on this framework, various correct-by-design controller synthesis approaches ~\cite{belta2017formal,yin2024formal,kloetzer2007temporal,chen2011formal,liu2025controller,guo2015multi,wang2025correct} have been developed to construct provably-correct controllers that satisfy formal specifications for both single-agent and multi-agent systems.

Despite these advances, standard Temporal Logics such as Linear Temporal Logic (LTL) have limitations in expressing collective behaviors of MAS. When the number of agents increases, the standard LTL is either too complex in structure or fails to accurately express group-level tasks. Counting LTL (cLTL)~\cite{sahin2019multirobot, sahin2017provably} was therefore introduced to compactly specify collective behaviors in MAS. In particular, for group tasks that are independent of identities of agents and only rely on the number of agents, cLTL provides a highly expressive formalism. In~\cite{sahin2019multirobot, sahin2017provably}, Integer Linear Programming (ILP) based optimization approaches was proposed to synthesize trajectories for cLTL specifications. However, the computational costs of ILP grow rapidly with system complexity, and these works are developed for deterministic systems rather than stochastic models. In this paper, we aim at developing an efficient and scalable correct-by-design control synthesis approach for stochastic multi-agent systems under cLTL specifications. 

Existing approaches to correct-by-design synthesis for stochastic systems typically construct finite Markov decision process (MDP) as abstractions of continuous-state systems, with stochastic simulation relations providing certified bounds between the abstract model and the original system~\cite{haesaert2020robust,abate2008probabilistic,schon2023verifying,wang2025correct,lavaei2022automated}. However, as the system dimension and scale increase, the size of the resulting abstract MDP grows rapidly, leading to exponential increases in computational and memory complexity. This well-known \emph{curse of dimensionality} poses a major challenge for correct-by-design synthesis under temporal logic specifications. Although several works \cite{haesaert2020robust,schon2023verifying,lavaei2022automated} have attempted to mitigate this issue, they mainly improve the scalability of abstraction construction rather than the computational efficiency of the synthesis algorithm itself.

% Prior approaches investigate correct-by-construction control synthesis of stochastic systems by abstracting the continuous nature of the state (and possibly action) spaces as finite Markov Decision Processes (MDPs), with stochastic simulation relations providing certified bounds that link abstract and original models.  \cite{haesaert2020robust,abate2008probabilistic,schon2022correct, schon2023verifying,wang2025correct,lavaei2022automated}
% In this paper, we consider abstracting stochastic systems discretely as Markov Decision Processes (MDPs)~\cite{schon2022correct, schon2023verifying}. 
%Given cLTL to quantify collective behaviors, we combine MDPs with a product automata to model and solve the correct-by-design control problem. It is obvious that the state space of abstracted MDP grows correspondingly as the scale of the system increases, leading to exponential growth in both computational and memory costs. This is well-known as the curse of dimensionality, causing the control synthesis computationally intractable.

To overcome this limitation, prior work in dynamic programming and temporal-logic control has explored low-rank tensor approximations to obtain compact value-function representations and accelerate value iteration~\cite{gorodetsky2018high,rozada2024tensor,ong2015value,alora2016automated}. These methods are motivated by the observation that value functions often exhibit low intrinsic dimensionality, which can be exploited through tensor decompositions to achieve computational savings~\cite{kolda2009tensor,eckart1936approximation}.
More recently, \cite{wang2025unraveling} proposed a low-rank value iteration framework for correct-by-design control synthesis of stochastic systems under syntactically co-safe LTL (scLTL) specifications. This framework integrates a deterministic finite automaton (DFA)-informed dynamic programming operator with a tree-based value iteration scheme to provably control tensor rank, guarantees a lower bound on the specification satisfaction probability, and significantly improves the computational and memory efficiency of dynamic programming for high-dimensional systems. Nevertheless, that approach is limited to scLTL specifications and still suffers from considerable redundant computation.
In particular, as temporal specifications become more complex, the induced DFA typically grows in both the number of states and outgoing transitions, which increases tensor-rank growth and tree size, thereby limiting scalability to more complex specifications that are inherent to multi-agent settings.
In this paper, we build upon the framework of~\cite{wang2025unraveling} and propose an improved optimization algorithm that can extend to large-scale homogeneous MAS under a class of counting LTL specifications, while further reducing both computational and memory costs.
To the best of our knowledge, this is the first paper dealing with correct-by-construction synthesis of counting LTL specifications for stochastic multi-agent systems.

%Nevertheless, existing tensor-based methods for general dynamic programming and temporal-logic control are not directly applicable to the computation of probabilistic safety guarantees for stochastic systems.

%To address this issue, an intuitive idea is to avoid directly dealing with high-dimensional systems. Under the MDP framework, a common approach is to conduct compact value approximations through decompositions. 
%Neural network–based methods~\cite{zhao2014neural, hure2021deep} have been applied to the approximation of value functions in optimal control, but these methods typically rely on assumptions about the approximation capability of neural networks. 
%In the context of tensor approximation, low-rank optimization has been studied~\cite{sidiropoulos2017tensor, kolda2009tensor}. Building on this, exploiting the characteristic that value functions typically have a low intrinsic dimensionality~\cite{eckart1936approximation}, \cite{ong2015value,alora2016automated} utilized low-rank tensor representations of value functions to implement value iteration based on dynamic programming, achieving computational savings. 

%More recently, \cite{wang2025unraveling} proposed a DFA-informed dynamic programming operator that tightly integrates temporal logic specifications with stochastic systems, and developed a tree-based value iteration approach that guarantees a lower bound of satisfaction probability while providing optimized control policies with substantially reduced computation.

\par 
\noindent {\bfseries Contributions.} (1) We show that, for decoupled homogeneous stochastic MAS under cLTL specifications, the satisfaction probability of which can be represented as a sum of rank-$1$ tensors induced by labeled paths of the DFA,  with each path admitting an agent-wise decomposition. (2) Based on this decomposition, we propose a dual-tree structure for approximate value iteration that compresses the required memory. (3) We demonstrate the effectiveness and scalability of our approach through large-scale case studies, showing significant reductions in both the memory usage and runtime compared with \cite{wang2025unraveling}. 
%\sofie{To be written}

\noindent {\bfseries Organization.} This paper is organized as follows. Sec.~\ref{sec:pre} introduces the preliminaries and problem setup. Sec.~\ref{sec:SPVI} presents the computation of satisfaction probability of cLTL with tensor-based value function, and Sec.~\ref{sec:tree} proposes a dual-tree structure to iteratively compute the value. Sec.~\ref{sec:case} demonstrates the effectiveness of the method through case studies and Sec.~\ref{sec:conclusion} concludes the paper and discusses potential future work.

\section{Homogeneous multi-agent MDPs with cLTL}% with Counting Linear Temporal Logics } %Preliminaries and Problem Statement}
\label{sec:pre}
% We first provide the necessary preliminaries used throughout this work and then formulate the problem to be addressed. Specifically, we model the system as a Markov Decision Process (MDP), introduce counting Linear Temporal Logic (counting LTL) for multi-agent specification, and state the problem under certain assumptions. \sofie{paragraph can be shortened by removing unnecessary words}\ruohan{I don't think any of it is necessary so I commented all out}

\subsection{Notations}
For a vector $\mathrm{v}\in\mathbb{R}^n$, we define its $1$-norm as $\|\mathrm{v}\|_1:=\sum_{j=1}^n |\mathrm{v}_j|$. For vectors $\mathrm{v}^i\in \mathbb{R}^{n_i}$ we define $\operatorname{col}(\mathrm{v}^1,\ldots,\mathrm{v}^N):=[\mathrm{v}^{1\top}\ldots \mathrm{v}^{N\top}]^\top \in\mathbb{R}^{\sum_{i=1}^N n_i}$. For vectors $\mathrm{v}^1\in \mathbb{R}^{n_1},\ldots,\mathrm{v}^N\in\mathbb{R}^{n_N}$, we define the outer product $\mathrm{v}^1\otimes \ldots \otimes \mathrm{v}^N \in \mathbb{R}^{\prod_{i=1}^N n_i}$ as a rank-$1$ tensor. We denote the set of nonnegative integers as $\mathbb{N}$. For $a,b\in\mathbb{N}$ with $a< b$, we denote the integer interval $\llbracket a,b \rrbracket:=\{a,a+1,\ldots,b\}$. For a set $S$, we denote its cardinality by $|S|$. For a given set $S$, let $S^* = \{\epsilon, s_0, s_0s_1, \ldots \mid s_i \in S\}$ denote the set of all finite sequences over $S$ (including the empty sequence $\epsilon$), and let $S^{\omega} = \{s_0 s_1 \ldots \mid s_i \in S\}$ denote the set of all infinite omega regular sequences over $S$. 

\subsection{Homogeneous Multi-Agent MDPs}
\begin{definition}[Markov decision process (MDP)]\hspace{1mm} An MDP is a tuple $\mdpM := (\mathbb{X}, \mathbb{A}, \mathbb{T}, \mathbb{X}_0)$, consisting of
\begin{itemize}
  \item a state space $\mathbb{X}$ with elements $x\in\spaceX$;
  % the finite set of possible states, where each state is denoted by $x \in \mathbb{X}$.
  \item an action space $\mathbb{A}$ with elements $a\in\spaceA$;
  % the finite set of possible actions, where each action is denoted by $a \in \mathbb{A}$.
  \item a stochastic kernel $\Tr:\spaceX \times \spaceA \times \spaceX \rightarrow [0,1]$, that assigns to each state-action pair $x\in \spaceX$ and $a\in \spaceA$ a probability distribution $\Tr(\cdot | x,a)$ over $\spaceX$;
  % \mathbb{T}(x^+ \mid x,a)$, 
  % the transition probability of moving to state $x^+$ when taking action $a$ in state $x$.
  \item a set of initial states $\mathbb X_0$ with initial state $x[0] \in \mathbb{X}_0$.
\end{itemize}
\label{eq:MDP}
\end{definition}
\noindent The \textit{execution (state trajectory)} of an MDP is $$\boldsymbol{x}_{[0,t]}:= \{ x[t]  \mid t=0,1,\ldots \}$$ initialized with the initial state $x[0]\in \spaceX_0$. In an execution, each consecutive state, $x[t+1]\in\spaceX$, is obtained as a realization $x[t+1] \sim \Tr(\cdot \mid x[t],a[t])$ of the stochastic kernel. We say that an MDP is finite if both $\spaceX$ and $\spaceA$ are finite. In this paper, we consider a \textit{homogeneous} multi-agent system of $N$ dynamically-decoupled agents, where homogeneity refers to all agents sharing an identical state space $\spaceX_c$, action space $\spaceA_c$, and stochastic kernel $\Tr_c$ depending solely on each agent's own state and action:
$$\mdpM^i:=(\spaceX_c,\spaceA_c,\Tr_c,\mathbb X_{c,0}), \quad \forall i\in \setagent$$
\noindent with $x^i[0]\in\mathbb X_{c,0}$ the initial state of $i$-th agent, $\setagent:=\llbracket 1,N \rrbracket$ the set of agent indices. The multi-agent system is represented as the composition of $N$ such finite MDPs:
\begin{align} \label{eq:MAS}
    \mdpM=(\spaceX,\spaceA,\Tr,\mathbb X_0 )
\end{align}
with $\spaceX=\prod_{i=1}^{N}\spaceX_c$, $\spaceA=\prod_{i=1}^{N}\spaceA_c$, $\Tr=\prod_{i=1}^{N}\Tr_c$, and $\mathbb X_0 :=\prod_{i=1}^{N} \mathbb X_{c,0}$. 
%x[0]=\operatorname{col}(x^1[0]\ldots x^N[0])$. 
%The \textit{execution (state trajectory)} of the system \eqref{eq:MAS} is 
% $\boldsymbol{x}_{[0,t]}:=\{x[t] | t=0,1,\ldots\}$ initialized with $x_0\in \spaceX_0$. Consecutive states $x[t+1]$ are obtained from realizations $x[t+1]\sim \Tr(\cdot| x[t],a[t])$ of the controlled stochastic kernel.  

% In an execution, each consecutive state $x[t+1]\in\spaceX$, is obtained as a realization $x_{t+1}\sim \Tr(\cdot|x,a)$ of the stochastic kernel. 
\begin{definition}[Markov policy $\bm{\pi}$]\, 
A Markov policy $\bm{\pi}$ is a sequence $\bm{\pi} = (\pi[0], \pi[1], \pi[2], \dots)$ where $\pi[t] : \mathbb{X} \rightarrow \mathbb{A}$ maps every state $x \in \mathbb{X}$ to an action $a \in \mathbb{A}$.
\label{def:markovpol}
\end{definition}
A Markov policy $\bm{\pi}$ is \emph{stationary} if $\pi[t]$ does not depend on the time index $t$, 
that is, for all $t$ it holds that $\pi[t] := \pi$ with $\pi : \mathbb{X} \to \mathbb{A}$. More general control strategies may depend on the full history. In this paper, we restrict to finite-memory control strategies depending on the specification, the definition of which follows in the next section.

% \sofie{Define MAS via MDPs}\ruohan{DONE}
% \ruohan{[\textbf{Ruohan: More clarification on structure of policy and connection between $\boldsymbol{\pi}$ and $\mathbf{C}_{\boldsymbol{\pi}}$ are needed here.}]}

\subsection{Counting Linear Temporal Logics}
\label{subsec:cLTL}
We associate a common set of atomic propositions $\AP_c:=\{p_1,\ldots,p_{|\AP_c|}\}$ to all agents $\mdpM^i, i\in \setagent$. The labeling map $L_c:\spaceX_c \rightarrow 2^{\AP_c}$ gives the set of true atomic propositions for $x^i\in \mathbb X_c$, that is, $L_c(x^i)\subset \AP_c$.  For the multi-agent system, we extend this labeling as $L:\spaceX\rightarrow \prod_{i=1}^N 2^{\AP_c}$ with 
$$L(x):=\operatorname{col}(L_c(x^1), \ldots, L_c(x^N)).$$
Let $\Sigma:=\prod_{i=1}^N 2^{\AP_c}$ denote the alphabet over which each labeling function $L(x)$ takes values, i.e., $l=L(x)\in\Sigma$. 

For a given state trajectory $\boldsymbol{x}=x[0],x[1],\ldots$ of the system~\eqref{eq:MAS},  each state can be translated to a specific letter $\letter[t]=L(x[t])$ using the labeling map $L$. Translating all states in $\boldsymbol{x}$ we can obtain the infinite word $\word:=\letter[0],\letter[1],\ldots$. Similarly, state trajectory suffixes denoted as $\boldsymbol{x}_t:=x[t],x[t+1],\ldots$ can be translated into word suffixes $\word_t:=\letter[t],\letter[t+1],\ldots$.

\medskip

As we are interested in the number of agents that satisfy a given atomic proposition as in \cite{sahin2019multirobot}, we define \textit{counting propositions} $[p,m]\in \AP_c\times \mathbb N$ that are satisfied by a given letter $l\in \prod_{i=1}^N 2^{\AP_c}$ iff  $|\big\{i\in \setagent: p\in l^i\big\}|  \geq m$. Combining counting propositions with temporal operators, we introduce counting Linear Temporal Logics based on the following syntactically co-safe syntax fragment adapted from \cite{sahin2019multirobot,sahin2017provably}. 
\begin{definition}[cLTL syntax]\hspace{1mm}
\label{def:cltl}
An sc-cLTL formula $\mu$ over a set of counting propositions $\AP$ with $[p,m]\in \AP$ has syntax 
$$\mu ::=[p,m] \,|\, \neg[ p,m] \,|\,  \mu_1 \wedge \mu_2 \,|\, \mu_1 \lor \mu_2 \,|\, \bigcirc \mu \,|\, \mu_1  \until \mu_2.$$
\end{definition}
The semantics of the syntax can be given for the suffixes $\word_t$. A counting proposition $\word_t \models [p,m]$ holds if $|\big\{i\in \setagent: p\in l[t]^i\big\}|  \geq m$,
%$[p,m]\in \letter_t$,
while a negation $\word_t \models \neg[ p,m]$ holds if $\word_t \not\models [ p,m]$. Furthermore, a conjunction $\word_t \models \mu_1 \wedge \mu_2$ holds if both $\word_t \models \mu_1$ and $\word_t \models \mu_2$ are true, while a disjunction $\word_t \models \mu_1 \lor \mu_2$ holds if either $\word_t \models \mu_1$ or $\word_t \models \mu_2$ is true. Also, a next statement $\word_t \models \bigcirc \mu$ holds if $\word_{t+1} \models \mu$. Finally, an until statement $\word_t \models \mu_1 \until \mu_2$ holds if there exists an $i\in\mathbb{N}$ such that $\word_{t+i} \models \mu_2$ and for all $j \in \mathbb{N}, 0 \leq j < i$, we have $\word_{t+j} \models \mu_1$. An eventual statement $\eventually \mu:= \text{True} \until \mu$ is a derived operator, where $\word_t \models \eventually \mu$ holds if there exists an $k\in \mathbb{N}$ such that $\word_{t+k} \models \mu$.

\begin{example}\hspace{1mm}
Consider the sc-cLTL formula $\mu = \Diamond[p, N/2]$, where $p\in\AP_c$ is an atomic proposition. The counting proposition $[p,N/2]$
% Here, $p$ is an atomic proposition and $[p, N/2]$ 
requires at least $N/2$ agents to satisfy $p$ simultaneously, and the formula $\mu$ specifies that eventually at least $N/2$ agents satisfy $p$.
\end{example}

Note that the above cLTL definition (Def. \ref{def:cltl}) is restricted to the syntactically co-safe fragment of cLTL  defined in \cite{sahin2019multirobot,sahin2017provably}. %The co-safe fragment is the class of  formulas for which every satisfying infinite word $\word\models \mu$ has a finite prefix $\word_[0,t]\models  \mu$ such that 
The co-safe fragment is the class of formulas for which every satisfying infinite word 
$\word \models \mu$ has a finite prefix $\word_{[0,t]}$ that witnesses the satisfaction of the formulae. 

Given a finite prefix $\word_{[0,t]} \in \Sigma^*$ and an infinite suffix 
$\word_{t+1} \in \Sigma^{\omega}$, their concatenation yields a new infinite word 
$\word_{[0,t]} \cdot \word_{t+1} \in \Sigma^{\omega}$, where $\letter[i]$ for $i \leq t$ 
are determined by the prefix and $\letter[t+1], \letter[{t+2}], \ldots$ are determined 
by the suffix $\word_{t+1}$.
A witness $\word_{[0,t]}$ is a finite prefix
such that 
$\word_{[0,t]} \cdot \word_{t+1} \models \mu$ for all $\word_{t+1} \in \Sigma^{\omega}$, 
where $t$ is minimal, i.e., 
$
    t = \min\{ t' \in \mathbb{N} \mid \word_{[0,t']} \cdot \word_{t'+1} \models \mu,\ 
    \forall \word_{t'+1}\in \Sigma^{\omega} \}.
$ The set of all finite witnesses is defined as $\setwitness:=\{\word_{[0,t]} \in \Sigma^* | t\in\mathbb{N}, \word_{[0,t]} \models \mu\}$. In this paper, we consider control strategies defined as follows.
\begin{definition}[Control strategies $\mathrm{C}_{\boldsymbol{\pi}}$] \hspace{1mm}
A control strategy
$\mathrm{C}_{\boldsymbol{\pi}}:=\spaceX \times \Sigma^* \rightarrow \spaceA$
assigns an action $a\in\spaceA$ to each pair $(x,\word_{[0,t]})\in \spaceX \times \Sigma^*$.
% , where $\word_{[0,t]}$ is the finite prefix up to time $t$.     
\end{definition}
A control strategy $\mathrm{C}_{\boldsymbol{\pi}}$ is said to be decoupled if it admits a factored representation: 
\begin{tightdisplay}
$$\mathrm{C}_{\boldsymbol{\pi}}^{\text{dec}}(x,\word_{[0,t]})=\operatorname{col}(\mathrm{C}_{\boldsymbol{\pi}}^1(x^1,\word_{[0,t]}),\ldots,\mathrm{C}_{\boldsymbol{\pi}}^N(x^N,\word_{[0,t]})),$$
\end{tightdisplay}
where the strategy for each agent $\mathrm{C}_{\boldsymbol{\pi}}^i:\spaceX_c\times \Sigma^* \rightarrow \spaceA_c$ depends only on its own state $x^i\in\spaceX_c$ and the shared prefix $\word_{[0,t]}$.

\subsection{Problem Statement}

\label{sec:problem} Given a stochastic homogeneous multi-agent system \eqref{eq:MAS} and an sc-cLTL specification $\mu$, synthesize a control strategy $\mathrm{C}_{\boldsymbol{\pi}}$ and probability guarantee $p_\mu$ such that the specification satisfaction probability of the controlled system initialized as $x_0$ satisfies:
\begin{tightdisplay}
\begin{equation}\label{eq:obj}
    \mathbb P^{x_0}\big(\mdpM \times \mathrm{C}_{\bm{\pi}} \models \mu\big)\geq p_\mu
\end{equation}
\end{tightdisplay}
while exploiting the homogeneity of the multi-agent system to mitigate the computational complexity that grows exponentially with the number of agents.
\section{Satisfaction probability of cLTL}
\label{sec:SPVI}
In this section, we show that the satisfaction probability of an sc-cLTL specification can be computed via leveraging a witness tree in which each vertex is associated with a witness and carries its rank-$1$ tensor probability. We first decompose the satisfaction probability into a sum of witness probabilities, each factorizing over agents. We then introduce a DFA to constrain the policy memory, and show how a witness tree exploits this decomposition.
\subsection{Rank-$1$ Decomposition of Witness Probabilities}
Owing to the co-safety property of $\mu$, the specification satisfaction is witnessed by a finite prefix of the generated word, and such prefixes are exactly the words accepted by the DFA \cite{vasile2017minimum}. The satisfaction probability in \eqref{eq:obj} is equivalent to the probability that the word $\word$ generated by $\mdpM \times \mathrm{C}_{\boldsymbol{\pi}}$ contains a finite witness $\word_{[0,t]}\in\setwitness$, that is
$$ \mathbb P\big(\mdpM \times \mathrm{C}_{\bm{\pi}} \models \mu\big)=\mbP_{\mdpM \times \mathrm{C}_{\boldsymbol{\pi}}}(\exists t\in \mathbb{N} : \word_{[0,t]} \in \setwitness),$$
where $\word_{[0,t]}$ denotes the finite prefix of the word $\word\in \Sigma^{\omega}$ generated by $\mdpM\times \mathrm{C}_{\boldsymbol{\pi}}$ up to time $t$. This probability can be lower-bounded by a finite-horizon probability for any $T\in\mathbb{N}$:
$$\mbP_{\mdpM \times \mathrm{C}_{\boldsymbol{\pi} }}(\exists t \in\! \mathbb{N} \!:\! \word_{[0,t]}\in \setwitness)\!\geq\! \mbP_{\mdpM \times \mathrm{C}_{\boldsymbol{\pi} }}(\exists t \leq\! T\!:\! \word_{[0,t]}\in \setwitness ).$$

Since each word $\word\in\Sigma^{\omega}$ generated by $\mdpM \times \mathrm{C}_{\boldsymbol{\pi}}$ contains at most one minimal witness $\word_{[0,t]}\in\setwitness$, % witnesses with distinct $t$ are mutually exclusive, and the finite-horizon satisfaction probability decomposes as a sum of probabilities over $\setwitness$, as stated in the following lemma.
 the set of witnesses defines a union of disjoint events.  Therefore the finite-horizon satisfaction probability decomposes as a sum of probabilities over $\setwitness$, as stated in the following lemma.

\begin{lemma}[sc-cLTL satisfaction]\label{lem:sum}\hspace{0.5mm}
The probability of the controlled system $\mdpM \times \mathrm{C}_{\boldsymbol{\pi}}$ initialized as $x_0\in \spaceX_0$ satisfying an sc-cLTL specification $\mu$ within time horizon $T$ decomposes as
\begin{tightdisplay}
\begin{equation}
    \mbP^{x_0}_{\mdpM \times \mathrm{C}_{\boldsymbol{\pi} }}(\exists t \leq T: \word_{[0,t]}\in \setwitness )=\sum_{\word_{[0,t]}\in\setwitness, t\leq T} P_{\word_{[0,t]}}^{\mathrm{C}_{\boldsymbol{\pi} }} (x_0),
    \label{eq:summation}
\end{equation}
\end{tightdisplay}
where $P_{\word_{[0,t]}}^{\mathrm{C}_{\boldsymbol{\pi} }}(x_0)$ is the probability that the controlled system $\mdpM \times \mathrm{C}_{\boldsymbol{\pi}}$ generates the finite prefix $\word_{[0,t]}=l[0],\ldots,l[t]$ as a witness in $\setwitness$, that is
$$P_{\word_{[0,t]}}^{\mathrm{C}_{\boldsymbol{\pi} }} (x_0)\!:=\!\mathbb{P}^{\mathrm{C}_{\boldsymbol{\pi} }}(L(x[k])\!=\! l[k],\forall k \! \in \! \llbracket 0,t \rrbracket|x[0]=x_0 ).$$
\end{lemma}

For a finite witness $\word_{[0,t]}=\letter[0],\ldots,\letter[t]$, we denote the projection of it on the $i$-th agent as $\word^i_{[0,t]}:=\letter[0]^i,\ldots,\letter[t]^i$, where $\letter[k]^i=L_c(x^i[k])$ for $k\in \llbracket 0,t \rrbracket$.

\begin{proposition}\label{prop:productP}\hspace{1mm}
Consider the homogeneous multi-agent system \eqref{eq:MAS} and a decoupled control strategy $\mathrm{C}_{\boldsymbol{\pi}}^{\text{dec}}$. Each witness probability $P_{\word_{[0,t]}}^{\mathrm{C}_{\boldsymbol{\pi} }}(x_0)$ in Lemma~\ref{lem:sum} factorizes over agents as 
\begin{equation}
P_{\word_{[0,t]}}^{\mathrm{C}^{\text{dec}}_{\boldsymbol{\pi} }}(x_0) = \prod_{i\in \setagent}
    P_{\word^i_{[0,t]}}^{\mathrm{C}^i_{\boldsymbol{\pi} }}(x^i_0)
    \label{eq:rank1}
\end{equation}
with the per-agent witness probability  
$$P_{\word^i_{[0,t]}}^{\mathrm{C}^i_{\boldsymbol{\pi} }} (x^i_0)\!:=\!\mathbb{P}^{\mathrm{C}^i_{\boldsymbol{\pi} }}(L_c(x^i[k])\!=\! l[k]^i,\forall k \! \in \! \llbracket 0,t \rrbracket|x^i[0]=x^i_0 ).$$ 
\end{proposition}

Let $P_{\word^i_{[0,t]}}^{\mathrm{C}^i_{\boldsymbol{\pi} }} \in \mathbb{R}^{|\spaceX_c|}$ denote the vector of per-agent witness probabilities and $P_{\word_{[0,t]}}^{\mathrm{C}^{\text{dec}}_{\boldsymbol{\pi} }}\in \mathbb{R}^{\prod_{i\in\setagent} |\spaceX_c|}$ the tensor of joint witness probabilities. Given the joint state space $\mathbb{X}=\prod_{i\in\setagent}\mathbb{X}_c$, the point-wise product in \eqref{eq:rank1} corresponds to an outer product, yielding a rank-$1$ tensor:
% the witness probability over the initial state set $\spaceX_0$ can be represented by a tensor 
% $P_{\word_{[0,t]}}^{\mathrm{C}_{\boldsymbol{\pi} }}\in \mathbb{R}^{\prod_{i=1}^N \spaceX_{c}}$. It follows from Proposition~\ref{prop:productP}, each witness probability tensor admits a rank-$1$ decomposition: 
$P_{\word_{[0,t]}}^{\mathrm{C}^{\text{dec}}_{\boldsymbol{\pi} }}=\bigotimes_{i\in\setagent} P_{\word^i_{[0,t]}}^{\mathrm{C}^i_{\boldsymbol{\pi} }}$. Consequently, by Lemma~\ref{lem:sum}, the finite-horizon satisfaction probability tensor is a summation of rank-$1$ tensors: 
\begin{equation}
    \mbP_{\mdpM \times \mathrm{C}^{\text{dec}}_{\boldsymbol{\pi} }}(\exists t \leq T: \word_{[0,t]}\in \setwitness )=\sum_{\word_{[0,t]},t\leq T} \bigotimes_{i\in\setagent} P_{\word^i_{[0,t]}}^{\mathrm{C}^i_{\boldsymbol{\pi} }}.
    \label{eq:ranksum}
\end{equation}
This rank-$1$ tensor decomposition reduces the per-witness memory cost from $O\left( \prod_{i=1}^{N}|\mathbb{X}_c|\right)$ to $O\left( \sum_{i=1}^{N}|\mathbb{X}_c|\right)$.

\subsection{DFA-constrained Control strategies}
 A deterministic finite automaton (DFA) is defined by the tuple 
$\mathcal{A} = (Q, q_0, \Sigma, \tau, q_f)$, where $Q$ denotes a finite set of states; $q_0\in Q$ is the initial state; $\Sigma$ is a finite alphabet; $\tau : Q \times \Sigma \rightarrow Q$ is a transition function; and $q_f$ is the accepting state. For a given word $\word= l[0],l[1],\ldots$, a run of a DFA defines a trajectory $q[0],q[1],\ldots$ initialized with $q[0]:=q_0$ and evolves according to 
$q[t+1] = \tau(q[t],\letter[t]).$ A finite word $\word_{[0,t]}\in\Sigma^*$ is accepted if the run of DFA ends at $q_f$. The set of finite words accepted by $\DFA$ is the finite-witness set $\setwitness$ defined in Section~\ref{subsec:cLTL}. Leveraging a DFA, we further define $\setwitness_q$ as the set of finite words whose induced run starting from $q$ reaches $q_f$.
Similar to \cite{kupferman2001model} for sc-LTL, we can state the following for sc-cLTL.
    For every sc-cLTL formula $\mu$ in Definition~\ref{def:cltl}, there exists a DFA $\DFA_\mu$ %with input alphabet $\Sigma_{\DFA_\mu}=2^{\AP}$, where $\AP$ is the set of counting propositions of $\mu$, 
    such that a word $\word$ satisfies $\mu$, that is, $\word\models \mu$, if and only if it has a finite prefix $\word_{[0,t]}$ in the language of $\mathcal{A}_\mu$. %such that generated by $\mdpM\times \mathrm{C}_{\boldsymbol{\pi}}$ is accepted by $\DFA_\mu$.

% As mentioned before, we are interested in finding a control strategy $C_{\boldsymbol \pi}$ for which we can find a non-trivial lower bound $p_\mu$. However, searching over the complete set of strategies that map from $\mathbb X\times \Sigma^\ast$ to $\mathbb A$ is infeasible since this includes policies that have infinite memory. Therefore, we consider policies whose memory is constrained to that of a chosen DFA $A_\mu=(Q, q_0, \Sigma, \tau, q_f)$ with per-agent policies
Searching over all control strategies mapping from $\spaceX \times\Sigma^*$ to $\spaceA$ is infeasible since $\Sigma^*$ is unbounded, requiring the strategy to condition on arbitrarily long prefixes. Therefore, we constrain the policy memory to a chosen DFA $A_\mu=(Q, q_0, \Sigma, \tau, q_f)$ with per-agent policies:
\begin{tightdisplay}
\begin{equation}
    \mathrm{C}_{\boldsymbol{\pi}}^i:\mathbb X_c \times Q \rightarrow \mathbb A_c\quad \textmd{ for } i\in \setagent
\end{equation}
\end{tightdisplay}based on which the desired decoupled control strategy $\mathrm{C}_{\boldsymbol{\pi}}^{\text{dec}}$ is defined. 
 Note that such a DFA, modelling $\mu$, is non-unique and increasing $|Q|$ allows each state to represent a smaller set of prefixes leading to that state. 
 % that by increasing the number of states in $Q$, each state can represent a smaller set of prefixes that lead up to that state.

\subsection{Value Iteration via Witness Trees}
We define value functions $\tensorV_{q}^{T}:\spaceX\rightarrow [0,1]$ for a given policy $\mathrm{C}_{\boldsymbol{\pi}}^{\text{dec}}$ as the satisfaction probability within time horizon $T$: 
% For a given policy $\mathrm{C}_{\boldsymbol{\pi}}^{\text{dec}}$, we can define a value function that expresses the satisfaction probability within time horizon $T$ as 
\begin{equation}
\tensorV_{q}^T(x)   = \mathbb{P}^{x}_{\mdpM \times \mathrm{C}_{\boldsymbol{\pi}}^{\text{dec}}}(\exists t \leq\! T:\word_{[0,t] }\in \setwitness_q| x[0]=x) . 
\end{equation} 
The satisfaction probability for an initial state $x_0\in \spaceX_0$ is then given by
% For a given initial state $x_0\in\mathbb X_0$, we can compute the satisfaction probability based on this value function as
\begin{equation}
\mathbb{P}^{x_0}_{\mdpM \times \mathrm{C}_{\boldsymbol{\pi}}^{\text{dec}}}(\exists t \leq\! T\!:\! \word_{[0,t]}\models\mu)=\tensorV_{\bar{q}_0}^{T}(x_0) 
\end{equation}with $\bar q_0= \tau_{\mathcal A}(q_0,L(x_0))$. Since the satisfaction probability decomposes over witnesses (Lemma~\ref{lem:sum}) with each term factorizing over agents (Proposition~\ref{prop:productP}), the value function $\tensorV_{q}^{T} (x)$ admits the rank-$1$ decomposition:
\begin{equation}
    \begin{aligned}
         \tensorV_{q}^T(x)   = \sum_{\word_{[0,t] }\in \setwitness_q,  t\leq T}
         \prod_{i\in\mathcal I}
         \mathbb{P}^{x^i}_{\mdpM^i \times \mathrm{C}_{\boldsymbol{\pi}}^{i}}(\word^i_{[0,t] }| x^i[0]=x^i).\label{eq:probabilities}
    \end{aligned}
\end{equation}

% The value function $\tensorV_{q}^T(x)$ inherits the structure of \eqref{eq:summation}
%  and \eqref{eq:rank1}, that is, 
%  \begin{align}
%      \tensorV_{q}^T(x)   &=\sum_{\word_{[0,t] }\in \setwitness_q, t\leq T}\mathbb{P}^{x}_{\mdpM \times \mathrm{C}_{\boldsymbol{\pi}}^{\text{dec}}}(\word_{[0,t] }| x[0]=x)
%      \end{align}
%      \begin{align}
%          \tensorV_{q}^T(x)   &= \sum_{\word_{[0,t] }\in \setwitness_q,  t\leq T}
%          \prod_{i\in\mathcal I}
%          \mathbb{P}^{x^i}_{\mdpM^i \times \mathrm{C}_{\boldsymbol{\pi}}^{i}}(\word^i_{[0,t] }| x^i[0]=x^i)\label{eq:probabilities}
%  \end{align}
To compute the value function $\tensorV_q^T$ in \eqref{eq:probabilities} efficiently, we associate each witness with a vertex of a tree and store its rank-$1$ probability as the vertex value, following \cite{wang2025unraveling}.
\begin{definition}[Witness tree $\mathcal{G}$ \cite{wang2025unraveling}] \label{def:r1tree}
\hspace{1mm}For a given DFA $\DFA=(Q,q_0,\Sigma,\tau,q_f)$, a witness tree
	$\mathcal{G}=(\mathcal{Z},\mathcal{E},\qmapping, \valuemapping)$ has \begin{itemize}
			\item a set of vertices $\mathcal{Z}$ with elements $z\in \mathcal{Z}$; 
			\item a set of labeled edges $\mathcal{E}$ with elements $ (z,\letter,z')\in\mathcal{E}$;  
			\item a DFA-state mapping $\qmapping:\mathcal{Z}\rightarrow Q$ that maps a vertex $z\in\mathcal{Z}$ to a DFA state $q\in Q$;
			\item a vertex value mapping $\valuemapping:\mathcal{Z}\rightarrow \mathbb{R}^{\prod_{i\in\setagent} |\spaceX_c|}$ that maps a vertex $z\in \mathcal{Z}$ to a rank-$1$ tensor: $\valuemapping(z)  = \bigotimes_{i\in\setagent} v^{i}(z).$
	\end{itemize}
\end{definition}
We call vertices with no outgoing edges \textit{leaf vertices}. The tree is rooted at $q_f$ and each edge $(z,l,z')$ connects a parent $z$ to its child $z'$. A path from a $q$-labeled vertex to the root defines a witness $\word_{[0,t]}\in \setwitness_q$ via its edges. Each product term in \eqref{eq:probabilities} corresponds to a $\bar{q}_0$-labeled vertex $z$, whose per-agent vector value $v^i(z)$ gives the per-agent witness probability $\mathbb{P}^{x^i}_{\mdpM^i \times \mathrm{C}_{\boldsymbol{\pi}}^{i}}(\word^i_{[0,t] }| x^i[0]=x^i)$. For a tree containing all witness of length at most $T$, the tensor value function $\tensorV_q^T$ is the sum of all rank-$1$ vertex values labeled $q$:
$$\tensorV_q^T=\sum_{z:\mathcal{L}_Q(z)=q} \bigotimes_{i\in\setagent} v^i(z).$$
The per-agent vector values $v^i(z)$ are propagated from parent to child via a \textit{per-agent operator} $\op_{\letter^i}^{\pi_q^i}(v^i)$ for a projected letter $\letter^i$ and per-agent policy $\pi^i_q:\spaceX_c \times Q \rightarrow \spaceA_c$:
\begin{equation} \label{eq:op_ori}
    \op_{\letter^i}^{\pi_q^i}(v^i)(x^i):=\expectation^{x^{i'}}[\mathcal{L}_{\letter^i}(x^{i'})v^i(x^{i'}) | x^i,a^i=\pi_q^i(x^i)],
\end{equation}
where $\mathcal{L}_{\letter^i}(x^{i'})=1$ if $L_c(x^{i'})\models \letter^i$, and $0$ otherwise.
\medskip

\par The witness tree (Definition~\ref{def:r1tree}) follows Definition~6 in \cite{wang2025unraveling}, where edge labels were Boolean subformulae grouping all letters that trigger the same DFA transition. For cLTL, each DFA transition is governed by counting propositions $[p,m]$, and the number of letters satisfying a counting proposition grows combinatorially with $N$. For instance, $[p,1]$ with $6$ agents is satisfied by $63$ distinct letters, each contributing to a new vertex per tree expansion step. This combinatorial growth has two consequences: (1) the witness tree becomes prohibitively large, and each vertex stores a tensor over the joint state space $\spaceX$, compounding the memory cost; (2) many of the resulting leaf vertices carry negligibly small probability values, wasting computation on insignificant probability contributions. 
\par
To address these challenges, in the next section, we label edges by individual letters and compress their enumeration via a Binary Decision Diagram (BDD)-based approach, exploit the homogeneity of the MAS through a dual-tree structure that avoids redundant per-agent computation, and prune low-probability vertices to bound tree growth. 

\section{Scalable control synthesis for cLTL via dual trees}
\label{sec:tree}
% \red{
% \begin{itemize}
%     % \item remove all time indexing  for policy in this subsection
%     % \item replace $\letter$ with $\letter$, no $\letter$ should exist 
%     % \item (together with next point) $\mathbb{P}^{x^i}_{\mdpM^i \times \mathrm{C}_{\boldsymbol{\pi}}^{i}}(\word^i_{[0,t] }| x^i[0]=x^i)=v^i_{\bar{q}_0}(z)$
%     % \item connect the tensor value, vector values back to \eqref{eq:probabilities}
%     % \item define value iteration via dfa-informed operator on $v^i(x^i)$ or $v(x)$
%     % \item reasoning on computational burden, give an example on that $[p,1]$ and say..
%     % \item start Section IV, move the previous stuff into Section III
%     % \item single-agent tree with relation
%     % \item remove time indexing in $\mathcal{T}$
%     % \item define single-agent operator
%     % \item Theorem 1 on `with dual tree, and algorithm ...' we have $\mbP_{\mdpM \times \mathrm{C}_{\boldsymbol{\pi} }}(\exists t \in \mathbb{N} : \word_{[0,t]}\in \setwitness)$ to be lower bounded by $\sum \prod_{i\in\setagent} W(R(z,i))$
% \end{itemize}
% }
In this section, we exploit the homogeneity of the MAS to eliminate redundant per-agent computation in the witness tree. We retain the witness-based tree topology via a multi-agent tree and offload the vertex value computation to a single-agent tree that stores only distinct per-agent values. Together, they form a dual-tree structure on which we perform value iteration and policy optimization. We further reduce computational costs via BDD-based letter compression and pruning of low-probability vertices.
% \par
% For homogeneous multi-agent systems, many per-agent witness probabilities $v^i(z)$ are identical across vertices and agents whenever they share the same agent-wise label and policy. To exploit this redundancy, we refer to Definition~\ref{def:r1tree} without the value mapping $v$ as a \textit{multi-agent tree} $\mathcal{G}_m:=(\mathcal{Z},\mathcal{E},\mathcal{L}_Q)$, and populate $v$ via a \textit{single-agent tree} $\mathcal{G}_s$ that stores each distinct per-agent witness probability only once, defined as follows. Together, $\mathcal{G}_m$ and $\mathcal{G}_s$ form the \textit{dual tree} structure. \ruohan{ref Fig.~\ref{fig:trees}}
\subsection{Dual-tree Value Iteration}
For homogeneous MAS, many per-agent witness probabilities $v^i(z)$ in the witness tree (Definition~\ref{def:r1tree}) are identical across vertices and agents whenever they share the same $\letter^i$ and per-agent policy. We define the \textit{multi-agent tree} $\mathcal{G}_{\mathrm{m}}:=(\mathcal{Z},\mathcal{E},\mathcal{L}_Q)$ as a witness tree without the value mapping $v$, and populate its vertex values via a \textit{single-agent tree} $\mathcal{G}_{\mathrm{s}}$, defined as follows. 

\begin{definition}[Single-agent tree]\label{def:single-agent-tree}\hspace{1mm}
Given a set of single-agent policies $\Pi_q$ with elements $\pi:\spaceX_c\rightarrow \spaceA_c$ for all $q\in Q$ and an alphabet set $2^{\AP_c}$, a single-agent tree is defined as a rooted graph $\mathcal{G}_{\mathrm{s}} = (\mathcal{K}, \mathcal{E}_{\mathrm{s}}, W)$, with 
\begin{itemize}
    \item a set of vertices $\mathcal{K}$ with elements $\kappa \in \mathcal{K}$;
    \item a set of labeled edges $\mathcal{E}_{\mathrm{s}}$ with elements $(\kappa,(\letter_{\mathrm{s}},\pi),\kappa')\in\mathcal{E}_{\mathrm{s}}$, where $(\letter_{\mathrm{s}},\pi)\in 2^{\AP_c} \times \Pi_q$ are the labels;
    \item a vertex value mapping $W:\mathcal{K} \rightarrow \mathbb{R}^{|\spaceX_c|}$ that maps a vertex $\kappa\in \mathcal{K}$ to a vector.
\end{itemize} 
\end{definition}
To populate the vertex values of $\mathcal{G}_{\mathrm{m}}$ from $\mathcal{G}_{\mathrm{s}}$, we define a relation between the two trees.
\begin{definition}[Relation between $\mathcal{G}_{\mathrm{m}}$ and $\mathcal{G}_{\mathrm{s}}$]\label{def:relation_tree} \hspace{1mm}
Given a multi-agent tree $\mathcal{G}_{\mathrm{m}} = (\mathcal{Z}, \mathcal{L}_Q, \mathcal{E})$ and a single-agent tree $\mathcal{G}_{\mathrm{s}} = (\mathcal{K}, \mathcal{E}_{\mathrm{s}}, W)$, we define a relation
\[
R \subseteq (\mathcal{Z} \times \setagent) \times \mathcal{K}
\]
that links each vertex-agent pair $(z,i) \in \mathcal{Z} \times \setagent$ to a unique vertex $\kappa \in \mathcal{K}$.
% denoted as $((z, i), \kappa) \in R$. 
\end{definition}
Since each $(z,i)$ maps to a unique $\kappa$, we denote the mapping $R(z,i):=\kappa$. %Note that the policy set $\Pi_q$ in Definition~\ref{def:single-agent-tree} is labeled by $q:=\mathcal{L}_Q(R(\kappa'))$, linking the edge labels of $\mathcal{G}_{\mathrm{s}}$ to the $q$-mode of vertices in $\mathcal{G}_{\mathrm{m}}$. 
The rank-$1$ tensor probability from Definition~\ref{def:r1tree} for the witness associated with vertex $z$ is then the outer product of per-agent vectors:
$$v(z)=\prod_{i\in\setagent} W(R(z,i)),$$
where each $W(R(z,i))\in\mathbb{R}^{|\spaceX_c|}$ gives the vector probability of the $i$-th agent for the witness associated with $z$. 
% tensor probability of the witness associated with $z$ is recovered as
% $$v(z)=\prod_{i\in\setagent} v^i(z) = \prod_{i\in\setagent} W(R(z,i)), $$
% where $v^i(z)=W(R(z,i))$ is the vector value at $z$ of the $i$th agent.
We extend $\mathcal{L}_Q$ to $\mathcal{K}$ by defining $\mathcal{L}_Q(\kappa):=q$ for any $(z,i)$ with $R(z,i)=\kappa$, which is well-defined since all such pairs share the same $q$ mode. The vertex values $W(\kappa)$ are propagated from parent $\kappa$ to child $\kappa'$ in $\mathcal{G}_{\mathrm{s}}$ via the \textit{single-agent operator} $\op^{\pi_q^{i}}_{l^{i}}(W(\kappa))$, which applies the per-agent operator \eqref{eq:op_ori} to the single-agent tree values, defined for a projected letter $\letter^i$ and per-agent policy $\pi:\spaceX_c \rightarrow \spaceA_c$ with $\pi\in \Pi_q, q=\mathcal{L}_Q(\kappa')$ as
\begin{equation}\label{eq:op_singleagent}
  \op^{\pi_q^{i}}_{l^{i}}(W\!(\kappa)\!)(x^{i}) \!:=\!  \mathbb{E}^{{x}^{i'}}[\mathcal{L}_{l^{i}}(x^{i'})W(\kappa)(x^{i'}) \!\mid \!x^{i},a^{i}\!=\!\pi_q^{i}(x^{i})], 
\end{equation}
where $\mathcal{L}_{\letter^i}(x^{i'})=1$ if $L_c(x^{i'})\models \letter^i$, and $0$ otherwise. Together, $\mathcal{G}_{\mathrm{m}}$, $\mathcal{G}_{\mathrm{s}}$ and $R$ form the \textit{dual-tree structure} (cf. Fig.~\ref{fig:trees}), initialized as $\mathcal{G}_{\mathrm{m}0}=( \{1\},\emptyset,\mathcal{L}_Q(1)=q_f)$ and $\mathcal{G}_{\mathrm{s}0}=( \{1\}, \emptyset,W_0  )$ with $W_0(1)=\boldsymbol{1}\in\mathbb{R}^{|\spaceX_c|}$, and $R=\{(1,i),1\}$ for all $i\in \setagent$.
\par
Given a horizon $T$ and a decoupled policy $\mathrm{C}_{\boldsymbol{\pi}}^{\text{dec}}$, we iteratively expand the trees and optimize $\mathrm{C}_{\boldsymbol{\pi}}^{\text{dec}}$. The trees are expanded as
\begin{align}
\label{eq:iter}
{\mathcal{G}_{\mathrm{m}}}_{k+1} = \mathcal{T}^{\pi}({\mathcal{G}_{\mathrm{m}}}_k),\hspace{2mm}{\mathcal{G}_{\mathrm{s}}}_{k+1} = \mathcal{T}^{\pi}({\mathcal{G}_{\mathrm{s}}}_k), \hspace{1mm} k\in \llbracket 0,T-1\rrbracket \nonumber
\end{align}
via Algorithm~\ref{alg:tree-single}, which simultaneously grows $\mathcal{G}_{\mathrm{m}}$ by adding child vertices for each letter-based DFA transition and $\mathcal{G}_{\mathrm{s}}$ by creating new vertices only when a distinct $(\letter^i,\pi)$ pair is encountered. The vertex values in $\mathcal{G}_{\mathrm{s}}$ are then updated bottom-to-top via Algorithm~\ref{alg:value-iter}, applying the single-agent operator along each edge. The following theorem establishes that value iteration on the dual tree lower bounds the specification satisfaction probability.
\begin{theorem}[Dual-tree value functions]\label{the1}\hspace{1mm}
Given a decoupled policy $\mathrm{C}_{\boldsymbol{\pi}}^{\text{dec}}$, a horizon $T$, and the trees $\mathcal{G}_{\mathrm{m}},\mathcal{G}_{\mathrm{s}}$ computed based on Algorithm~\ref{alg:tree-single} and \ref{alg:value-iter}, the specification satisfaction probability of the MAS initialized as $x_0\in\spaceX_0$ in problem statement is lower bounded as 
$$\mathbb P^{x_0}\big(\mdpM \times \mathrm{C}^{\text{dec}}_{\bm{\pi}} \models \mu\big) \geq  \sum_{z:\mathcal{L}_Q(z)=\bar{q}_0}  \prod_{ i\in \setagent}W(R(z,i))(x_0^i),$$
with $\bar{q}_0=\tau(q_0,L(x_0))$. 
\end{theorem}

\begin{figure}[t]
    \centering
    \includegraphics[width=\linewidth]{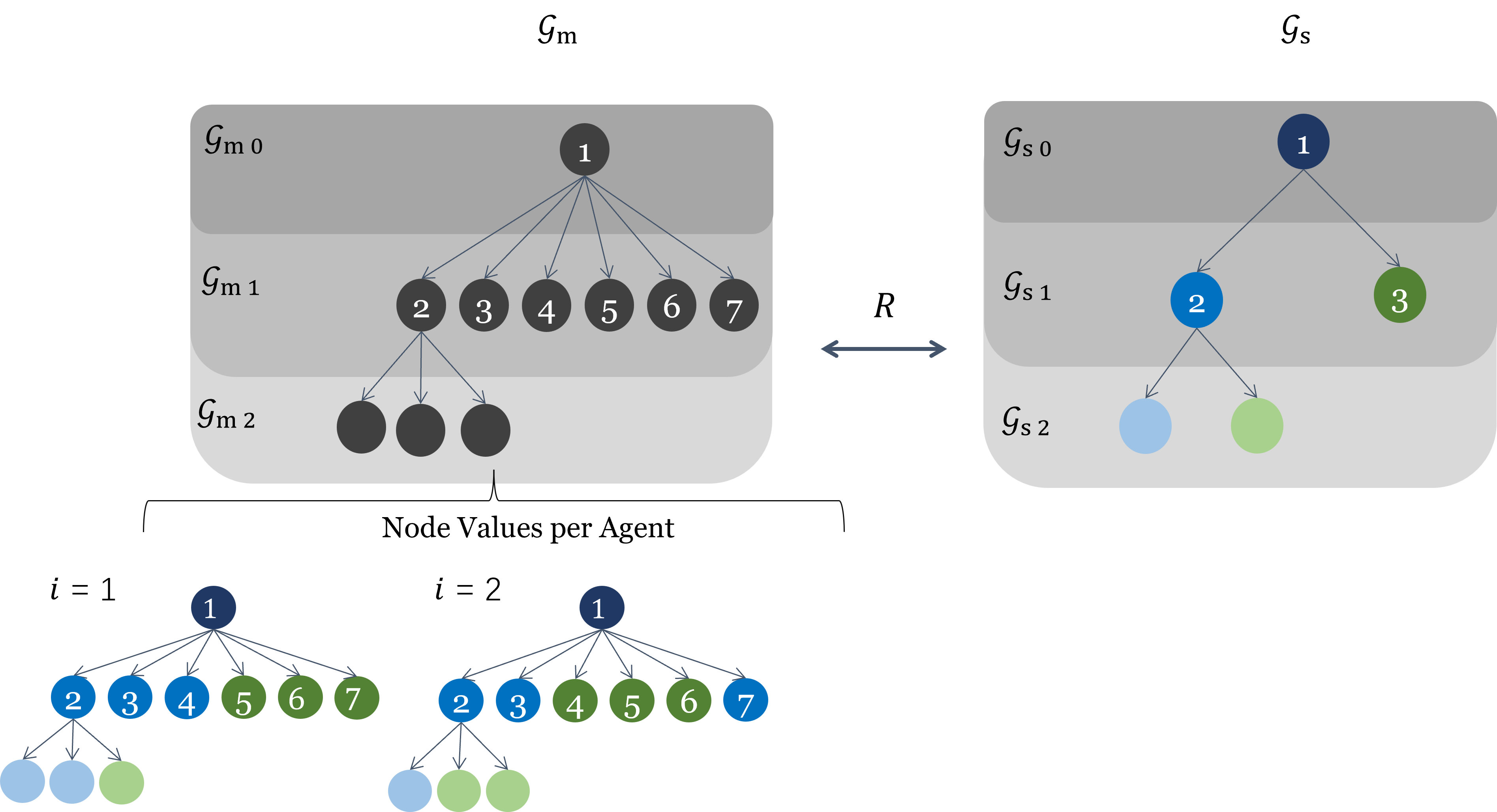}
    \caption{The multi-agent tree $\mathcal{G}_\mathrm{m}$ (top left) and single-agent tree $\mathcal{G}_\mathrm{s}$ (top right) grown simultaneously via Alg.~\ref{alg:tree-single}. As shown in "Node Values per Agent" (bottom), vertices sharing an identical value vector are indicated using the same color and are mapped to a unique vertex $\kappa \in \mathcal{K}$ in $\mathcal{G}_\mathrm{s}$ via $R$.}
    \label{fig:trees}
    \vspace{-0.3cm}
\end{figure}

\begin{algorithm}[htbp]
\caption{Growing Multi-Agent Tree $\mathcal{T}^{\pi}(\mathcal{G}_{\mathrm{m}})$ and Single-Agent Tree $\mathcal{T}^{\pi}(\mathcal{G}_{\mathrm{s}})$} 
\label{alg:tree-single}
\begin{algorithmic}[1]
\renewcommand{\algorithmicrequire}{\textbf{Input:}}
\renewcommand{\algorithmicensure}{\textbf{Output:}}
\Require  $\mathcal{G}_\mathrm{m}$, $\mathcal{G}_\mathrm{s}$ and $R$, DFA $\DFA$, decoupled policy $\mathrm{C}^{\text{dec}}_{\bm{\pi}}$.
\Ensure Expanded $\mathcal{G}_\mathrm{m}$ and $\mathcal{G}_\mathrm{s}$, updated $R$.
\ForAll{$z \in \mathcal{G}_{\mathrm{m}}.Leaves$} \Comment{Grow Multi-Agent Tree}
    \For{$\{(q,\letter,q') \in \tau_{\mathcal{A}} \mid q' = \mathcal{L}_Q(z)\}$}
    \State $\bar z \gets$ Create a new vertex for original tree
    \State $\mathcal{Z} \gets \mathcal{Z} \cup \{\bar z\}$
    \State $\mathcal{E} \gets \mathcal{E} \cup \{(z,l,\bar z)\}$
    \State $\mathcal{L}_Q(\bar z) \gets q$ 
    \ForAll{$i \in \setagent$} \Comment{Grow Single-Agent Tree}
      \State  $\pi(\cdot) \gets \mathrm{C}^{i}_{\bm{\pi}}(\cdot,q)$
      \State $\kappa \gets R(z,i)$    
      \State $\letter_\mathrm{s} \gets \letter^i$ 
      \If{$\exists \kappa' \in \mathcal{K}$ s.t. $\left(\kappa,(\letter_\mathrm{s}, \pi), \kappa' \right) \in \mathcal{E}_{\mathrm{s}}$}
        \State $R \gets R \cup \{((\bar z, i), \kappa')\}$ 
      \Else
        \State $\bar \kappa \gets$ Add new vertex to single-agent tree
        \State $\mathcal{K} \gets \mathcal{K} \cup \{\bar \kappa\}$
        \State $\mathcal{E}_{\mathrm{s}} \gets \mathcal{E}_{\mathrm{s}} \cup \{(\kappa,(\letter_\mathrm{s},\pi),\bar \kappa)\}$
        \State $R \gets R \cup \{((\bar z, i), \bar \kappa)\}$
      \EndIf
    \EndFor
  \EndFor
\EndFor
\end{algorithmic}
\end{algorithm}

% \begin{algorithm}[htbp]
% \caption{Value Iteration on single-agent tree $\mathcal{T}^\pi(\mathcal{G}_{\mathrm{sha}})$}
% \label{alg:value-iter}
% % \begin{algorithmic}[1]
% % \renewcommand{\algorithmicrequire}{\textbf{Input:}}
% % \renewcommand{\algorithmicensure}{\textbf{Output:}}
% % \Require $\mathcal{G}_\mathrm{ori}$, $\mathcal{G}_\mathrm{sha}$, DFA $\DFA$, policy $\pi_q$.
% % \Ensure Updated value vector $W$
% % \State $\mathcal{G}_{\mathrm{sha}}^+ \gets$ \textbf{Grow}$(\mathcal{G}_{\mathrm{sha}})$\Comment{Grow single-agent tree}
% %   \ForAll{$(\kappa, c, \kappa') \in \mathcal{G}_{\mathrm{sha}}^+.\mathcal{E}_{\mathrm{sha}}$ in bottom-to-top order}
% %     \State $W(\kappa') \gets
% %     \mathcal{T}^{\pi_q}_l(W(\kappa))$
% %   \EndFor
% % \end{algorithmic}
% % \end{algorithm}
% \begin{algorithmic}[1]
% \renewcommand{\algorithmicrequire}{\textbf{Input:}}
% \renewcommand{\algorithmicensure}{\textbf{Output:}}
% \Require $\mathcal{G}_\mathrm{ori}$, $\mathcal{G}_\mathrm{sha}$, DFA $\DFA$, policy $\pi_{q,c}$.
% \Ensure Updated value vector $W$
% \State $\mathcal{G}_{\mathrm{sha}}^+ \gets$ \textbf{Grow}$(\mathcal{G}_{\mathrm{sha}})$\Comment{Grow single-agent tree}
%   \ForAll{$(\kappa, c, \kappa') \in \mathcal{G}_{\mathrm{sha}}^+.\mathcal{E}_{\mathrm{sha}}$ in bottom-to-top order}
%     \State $(\kappa_{\mathrm{par}}, c, q) \gets \psi(\kappa')$
%     \State $W(\kappa') \gets \mathcal{T}^{\pi_{q,c}}_l(W(\kappa))$
%     %\State $W(\kappa') \gets
%     %\mathcal{T}^{\pi_q}_l(W(\kappa))$
%   \EndFor
% \end{algorithmic}
% \end{algorithm}
\begin{algorithm}[htbp]
\caption{Value Iteration on Dual Tree $\mathcal{G}_\mathrm{m}$, $\mathcal{G}_{\mathrm{s}}$}
\label{alg:value-iter}
\begin{algorithmic}[1]
\renewcommand{\algorithmicrequire}{\textbf{Input:}}
\renewcommand{\algorithmicensure}{\textbf{Output:}}
\Require $\mathcal{G}_\mathrm{m}$, $\mathcal{G}_\mathrm{s}$, DFA $\DFA$. %\xinyuan{decoupled policy $\mathrm{C}^{\text{dec}}_{\bm{\pi}}$}.
\Ensure Updated $W$
\State $\mathcal{G}_{\mathrm{m}}^+, \mathcal{G}_{\mathrm{s}}^+ \gets \mathcal{T}^\pi(\mathcal{G}_{\mathrm{m}}), \mathcal{T}^{\pi}(\mathcal{G}_{\mathrm{s}})$\Comment{Grow Trees}
  \ForAll{$(\kappa, (\letter_\mathrm{s},\pi), \kappa') \in \mathcal{G}_{\mathrm{s}}^+.\mathcal{E}_{\mathrm{s}}$ in bottom-to-top order}
    \State $W(\kappa') \gets \op^{\pi}_{\letter_\mathrm{s}}(W(\kappa))$
    %\State $W(\kappa') \gets
    %\mathcal{T}^{\pi_q}_l(W(\kappa))$
  \EndFor
\end{algorithmic}
\end{algorithm}

\par \noindent {\bfseries Policy Optimization.} Given $\mathcal{G}_{\mathrm{m}}$ and $\mathcal{G}_{\mathrm{s}}$, we optimize the per-agent policies $\pi_q^i\in\Pi_q$ for each $q\in Q$ as 
% For the $i$-th agent, the policy $\pi^i$ is updated as
% For each mode $q\in Q$, we develop an efficient heuristic optimization to obtain the policy set $\Pi_q$ with elements $\pi\in\Pi_q$
\vspace{-2mm}
 \begin{equation} 
 	\begin{aligned}
 	\pi^{i\ast}_q	\in \arg\max_{\pi^i_q} \sum_{e\in\mathcal E_{q}} \op_{\letter^{i}}^{	\pi^{i}_q }(W(R(z,i))) c_{e}^{i},
 	\end{aligned} \label{eq:pol_base}
 \end{equation}
where $\mathcal E_{q}:=\{e =(\kappa,\letter^i,\kappa')\in \mathcal E| \qmapping(\kappa') = q \}$ and $c^{i}_{e}:= \prod_{j\in \llbracket 1,N\rrbracket \setminus \{i\}} \|\op_{\letter^{j}}^{	\pi^{j}_q }(W(R(z,i)))\|_1.$ The scalar $c_e^i$ weights the contribution of the $i$-th agent at edge $e$ by the total probability mass of all other agents. For computational efficiency, we search for policies that are shared across all agents or subsets of agents, reducing the number of distinct per-agent operators and thereby limiting the width of the $\mathcal{G}_{\mathrm{s}}$.  

\par
\begin{myremark}\hspace{1mm}
    The witness-tree approach in \cite{wang2025unraveling} stores a value vector for every vertex-agent pair, requiring $|\mathcal{Z}| |\setagent|$ vectors. Our proposed dual-tree approach avoids this redundancy by storing each distinct per-agent value vector only once in $\mathcal{G}_{\mathrm{s}}$, reducing the number of stored vectors to $|\mathcal{K}|$ with $|\mathcal{K}| \ll |\mathcal{Z}|$.
    % Besides storing the tree structure itself, method that performs value iteration directly on the witness tree $\mathcal{G}$ require $\mathcal{O}(|\mathcal{Z}||\setagent||\mathbb{X}_c|)$ memory to represent the value tensors. For the proposed dual-tree structure, excluding the memory required for tree structure, which is nearly the same as witness tree, the memory cost consists of two parts, the storage of $W$ on single-agent tree $\mathcal{G}_\mathrm{s}$ with $|\mathcal{K}|$ nodes and the relation $R$, that is, $\mathcal{O}(|\mathcal{K}||\mathbb{X}_c| + |\mathcal{Z}||\setagent|)$. Since $|\mathcal{K}|\ll|\mathcal{Z}|$ and $|\setagent|\ll|\mathbb{X}_c|$, the dual tree structure is more memory-efficient than witness tree, especially when $|\mathcal{K}|/(|\mathcal{Z}||\mathcal{I}|)$ becomes smaller.
\end{myremark}

\subsection{Efficient Tree Compression}
%% Ruohan: I moved policy optimization to the previous subsection
% \subsection{Efficient Tree Compression and Policy Optimization}
 \label{subsec:BDD}
Though the single-agent tree significantly compresses the computation and hence reduces the memory overhead, the growth of the multi-agent tree still brings extra memory costs for storing the tree structure. To address this, we propose a way to group or bundle DFA transitions such that the tree becomes narrower. Additionally, we prune low-valued vertices in the tree.\\
\noindent {\bfseries Bundling DFA transitions for tree compression. }% rank-1 subformulae via BDDs .}
Up to now, we have used combinations of letters $l^i$ for each agent to define unique transitions in the DFA. However, as shown in \cite{wang2025unraveling}, this can be replaced by Boolean formulas that describe a set of letters. 
In \cite{wang2025unraveling}, we assumed that these subformula were given a-priori. For cLTL defined with counting propositions, we still need to construct these Boolean formulas. 
More precisely, we are interested in formulas that combine per-agent Boolean formula as follows
\begin{equation}\label{eq:boolean}
\alpha:=\bigwedge_{i\in\mathcal I} \alpha^i, \quad \alpha^i::= p | \neg \alpha^i |  \alpha_1^i\wedge \alpha^i_2.
\end{equation}

These formulae can still be computed with the single-agent operator in \eqref{eq:op_singleagent}, and lead to a narrower multi-agent tree.
To find appropriate formulae for the transitions, we symbolically represent the set of letters associated to a DFA transition with a BDD \cite{een2006translating}. 
Based on the automatic reordering of the BDD, we can find a set of Boolean formulae satisfying \eqref{eq:boolean}.

% Given the DFA $\mathcal{A} = (Q, q_0, \Sigma, \tau, q_f)$ with the transition $\tau: Q \times \Sigma \rightarrow Q$, we define the set of formulae as $ \alpha:= \{\alpha_1, \alpha_2, \dots, \alpha_{|\alpha|}\}$, in which each formula $\alpha := \{\alpha\}$
% For a letter $l \in \Sigma$, we say that there exists one $\alpha_i$ such that $l \models \alpha_i$, hence the transition of DFA states can be rewritten as $\tau: Q \times \alpha \rightarrow Q$. Given a DFA state $q$, each subformula $\alpha_i$ induces a subset of letters in $\Sigma$, denoted as $l \in B_{\alpha_i} \subset \Sigma$ iff $l \models \alpha_i$, and all letters in $B_{\alpha_i}$ trigger the same DFA transition from $q$. To obtain the explicit letters $l \in B_{\alpha_i} \subset \Sigma$, we encode the subformula $\alpha_i$ as a Boolean function and use BDD to compactly express the letters. For $\alpha_i$, each satisfying assignment represented by a path in the BDD yields a letter $l \in B_{\alpha_i}$. This BDD-based computation is performed as a preprocessing step before value iteration, reducing the cost of full enumeration by sharing common substructure \cite{een2006translating}. Hence we can grow a smaller multi-agent tree with fewer edges to compress the computation.
% \par

\noindent {\bfseries Dual-tree pruning.} The tree pruning approach removes leaf nodes whose corresponding values fall below a given threshold, as illustrated in Alg.~\ref{alg:prune}.
\begin{algorithm}[htbp]
\caption{Tree Pruning $\mathcal{P}(\mathcal{G})$}
\label{alg:prune}
\begin{algorithmic}[1]
\renewcommand{\algorithmicrequire}{\textbf{Input:}}
\renewcommand{\algorithmicensure}{\textbf{Output:}}
\Require $\mathcal{G}_\mathrm{m}$, $\mathcal{G}_\mathrm{s}$, $R$, and thresholds $\theta_{product}$, $\theta_{single}$
\Ensure Pruned $\mathcal{G}_\mathrm{m}$, $\mathcal{G}_\mathrm{s}$, $R$.
\For{$z \in \mathcal{G}_\mathrm{m}.\text{leaves}$}
    \State $score \leftarrow 1.0$
    \ForAll{$i \in \setagent$} 
        \State $\kappa \gets R(z,i)$
        \State $w_{max} \leftarrow \max(W(\kappa))$ 
        \If{$w_{max} < \theta_{single}$} \Comment{$<$ threshold $\theta_{single}$}
            \State $score \leftarrow 0.0$
            \State \textbf{break}
        \EndIf
        \State $score \leftarrow score \cdot w_{max}$
    \EndFor
    \If{$score < \theta_{product}$} \Comment{$<$ threshold $\theta_{product}$}
        \State $\mathcal{G}_\mathrm{m}.\mathcal{Z} \leftarrow \mathcal{G}_\mathrm{m}.\mathcal{Z} \setminus \{z\}$
        \State $\mathcal{G}_\mathrm{m}.\mathcal{E} \leftarrow \mathcal{G}_\mathrm{m}.\mathcal{E} \setminus \{\text{parent}(z), l, z\}$
        \State $R \gets R \setminus \{((z,i),\kappa)\}$
    \EndIf
\EndFor
\end{algorithmic}
\end{algorithm}
By Lemma~3 in \cite{wang2025unraveling}, the resulting sub-tree after pruning preserves a lower bound on the satisfaction probability, while enabling efficient tree-based value iteration. 
% \begin{myremark}\hspace{0.5mm}
    % Let $|\mathcal{Z}|$ and $|\hat{\mathcal{Z}}|$ denote the vertex count of $\mathcal{G}_m$ before and after BDD compression, with $|\hat{\mathcal{Z}}| \ll |\mathcal{Z}|$.
    % % Let $\mathcal{Z}$ denote the vertex set of $\mathcal{G}_{\mathrm{m}}$ without BDD compression and $\hat{\mathcal{Z}}$ the vertex set of $\mathcal{G}_{\mathrm{m}}$ after BDD compression, with $|\hat{\mathcal{Z}}| \ll |\mathcal{Z}|$. 
    % The method in \cite{wang2025unraveling} requires $\mathcal{O}(|\mathcal{Z}| |\setagent| |\spaceX_c|)$ memory. Our proposed dual-tree approach with BDD compression reduces this to $\mathcal{O}(|\mathcal{K}| |\spaceX_c| + |\hat{\mathcal{Z}}| |\setagent| )$ with $|\mathcal{K}| \ll |\hat{\mathcal{Z}}|$, where the first term accounts for the value mapping $W$ of $\mathcal{G}_{\mathrm{s}}$ and the second term for the relation $R$. 
    % % This yields a memory reduction ratio of $$1-\frac{|\mathcal{K}|}{|\mathcal{Z}| |\setagent|} -\frac{|\hat{\mathcal{Z}}|}{|\mathcal{Z}| |\spaceX_c|}.$$ 
    % Since computing with $|\mathcal{Z}|$ is infeasible in practice, the memory comparison in Section~\ref{sec:case} uses $\mathcal{Z}=\hat{\mathcal{Z}}$, so the reported reduction reflects solely the dual-tree gains. Tree pruning (Algorithm~\ref{alg:prune}) further reduces $|\hat{\mathcal{Z}}|$ and $|\mathcal{K}|$, though this reduction is case-dependent and cannot be quantified analytically.
% \end{myremark}

% \noindent {\bfseries Policy Optimization}

\section{Case Study}
\label{sec:case}
%In this section, we consider several case studies to show the effectiveness and scalability of our proposed approach. We compare our approach with the monolithic dynamic programming (DP) and the method in \cite{wang2025unraveling}, known as DFA Tree.
We consider a homogeneous MAS in which each agent evolves according to
\[
x^{i+} = x^i + u^i + w^i, \quad w^i\sim \mathcal{N}(0, 1), \quad \forall i =\llbracket 1, N \rrbracket,
\]
where $x=\operatorname{col}(x^1,\ldots,x^N)$ is the joint state, $u=\operatorname{col}(u^1,\ldots,u^N)$ the joint input, and $w_i$ an independent disturbance.
% where $x_i$ represents the height of drone i, $u_i$ is the control input, and $w_i$ denotes the environment noise. Each agent is a drone that moves only along a straight vertical line, that is, one-dimensional height motion. 
We abstract the MAS as a finite MDP following \cite{haesaert2020robust}, with $\mathbb{X}_c = [-10, 10]$ and $\mathbb{A}_c =[-2, 2]$. We demonstrate our method on four cLTL specifications. All experiments are conducted on a PC with a 13th Gen Intel Core i9-13900HX processor and 16.0 GB of RAM. 

\subsection{Comparison with existing methods}
We consider the following specification to compare the computational performance of monolithic dynamic programming, the witness-tree approach in \cite{wang2025unraveling}, and our dual-tree approach:
\[
\mu_1 = \neg [p_1, N/2] \;\mathsf{U}\; [p_2, N/3],
\]
for which labeling functions are defined as $L_c(x^i)=p_1$ if $x^i \in [2,4]$, $L_c(x^i)=p_2$ if $x^i \in [-4, -2]$. Each agent is abstracted into a finite MDP with $|\spaceX_c|=100$ states. Monolithic dynamic programming stores the value function over the full joint state space, and for $\mu_1$ its exponential memory growth exceeds the 16 GB RAM available on the PC. To validate our dual-tree approach, we report the lower bounds of satisfaction probabilities for three sets of initial conditions with $N=4$ agents in Tab.~\ref{tab:perform_results}. 
\par
We then compare the computational costs of the witness-tree approach as in \cite{wang2025unraveling} and the proposed dual-tree approach with a varying number of agents $N$, as shown in Fig.~\ref{fig:output_compare}. Our dual-tree approach significantly reduces the computational costs in terms of both memory usage and running time, especially in large-scale scenarios. The jumps in Fig.~\ref{fig:output_compare} are caused by the cLTL specification, whose number of letters grows combinatorially with $N$. It is observed that the memory reduction reaches nearly one order of magnitude at $N=8$, and beyond this point, the witness-tree method can no longer be executed as it runs out of memory. Moreover, the relative memory reduction of the dual-tree approach over the witness-tree method increases with $N$, demonstrating greater efficiency gains for larger-scale MAS.

% Although the computational costs of both approaches increases with the number of agents, our dual-tree approach achieves substantial savings, indicating improved scalability. 
\begin{figure}[t]
    \centering
    \includegraphics[width=0.7\linewidth]{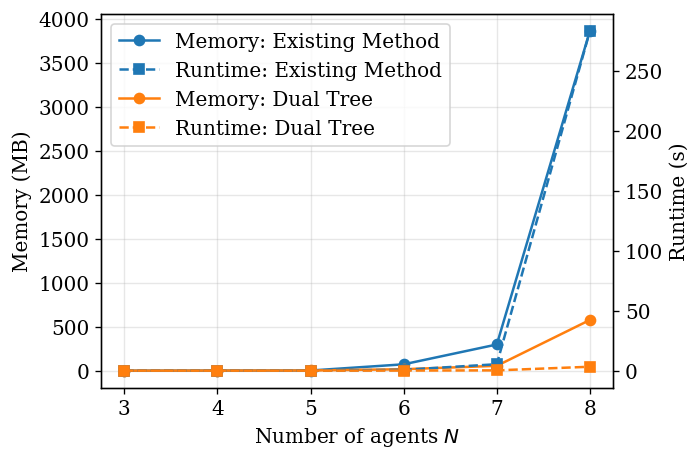}
    \caption{Comparison of existing method \cite{wang2025unraveling} and dual- tree approach in terms of memory usage and runtime for $\mu_1$ with an increasing number of agents.}
    \label{fig:output_compare}
     \vspace{-0.5cm}
\end{figure}

% \begin{tikzpicture}[x=0.38cm,y=1cm]
% \draw[thick,->,>=Stealth,black] (-10.5,0) -- (10.8,0);

% \foreach \x in {-10,-4,-2,0,2,4,10}{
%     \draw[thick,black] (\x,0.12) -- (\x,-0.12);
%     \node[below] at (\x,-0.12) {\small $\x$};
% }
% \foreach \x in {-2.1,-1.9,0.1,2.4}{
%     \node[red] at (\x,0) {\Large $\ast$};
% }
% \foreach \x in {-1.8,-1.7,1.8,1.7}{
%     \node[blue] at (\x,0) {\Large $\ast$};
% }

% \foreach \x in {2.3,1.0,1.5,0}{
%     \node[orange] at (\x,0) {\Large $\ast$};
% }
% \end{tikzpicture}

\begin{table}[htbp]
    \centering
    \caption{Lower bounds of satisfaction probabilities of $\mu_1$.} 
    \label{tab:perform_results}
    \begin{tabular}{cccc} 
        \toprule
        $p_\mu$ & $1.0$ & $0.6051$ & $0.3382$  \\ 
        \midrule
        $x^1[0]$ & $-2.1$ & $-1.8$ & $2.3$ \\
        $x^2[0]$ & $-1.9$ & $-1.7$ & $1.0$ \\
        $x^3[0]$ & $0.1$ & $1.8$ & $1.5$ \\
        $x^4[0]$ & $2.4$ & $1.7$ & $0$\\
        \bottomrule
    \end{tabular}
\end{table}

\subsection{Scalability of dual-tree value iteration}
To further show the scalability of our approach, we consider specifications as follows:
\begin{equation}
\begin{aligned}
    \mu_2 =
    &[p_1, N]\;
    \mathsf{U}\;
    ([p_2, 1] \land [p_1, N]),\\
    \mu_3 &= \bigwedge_{t=0}^{5}\bigcirc^t [p_1, N].
\end{aligned}
\nonumber
\end{equation}
The labeling functions for $\mu_2$ and $\mu_3$ are defined as follows. For $\mu_2$, $L_c(x^i) = p_1$ if $x^i \in [-5,5]$, and $L_c(x^i) = p_2$ if $x^i \in [-2,2]$. For $\mu_3$, $L_c(x^i) = p_1$ if $x^i \in [-5,5]$. For both $\mu_2$ and $\mu_3$, we abstract each agent using a finite MDP with $|\spaceX_c|=100$ states. The computational costs of the proposed dual-tree approach are evaluated with the number of agents ranging from $3$ to $18$. As shown in Fig.~\ref{fig:output_scalability34}, both memory usage and runtime increase approximately linearly with the number of agents $N$, demonstrating the scalability of our dual-tree approach.
%\siyuan{In particular, for $\mu_3$, they increase from $2.7393$~MB and $0.0088$~s at $N=3$ to $16.3769$~MB and $0.04778$~s at $N=18$. For $\mu_4$, they increase from $3.9242$ MB and $0.01260$ s to $23.2009$ MB and $0.0734$~s over the same range, Siyuan: this sentence can be removed if space is limited.} 

We then consider a specification with higher complexity: $$\mu_4 = [p_1, N/2] \;\land\; \Diamond [p_2, N/4]$$ with labeling functions $L_c(x^i) = p_1$ if $x^i \in [2,4]$, and $L_c(x^i) = p_2$ if $x^i \in [-4,-2]$. This added complexity arises because the intermediate counting thresholds in $\mu_4$ ($N/2$ and $N/4$) produce far more satisfying letter assignments than the extreme thresholds ($N$ and $1$) in $\mu_2$ and $\mu_3$. As shown in Fig.~\ref{fig:output_scalability5}, though the memory usage and runtime grow more rapidly with the number of agents, which is caused by the combinatorial explosion in the number of BDD-generated letters as $N$ increases, the dual-tree method remains computationally feasible, demonstrating its ability to handle specifications with higher complexity.  

 \begin{figure}[t]
    \centering
    \includegraphics[width=0.7\linewidth]{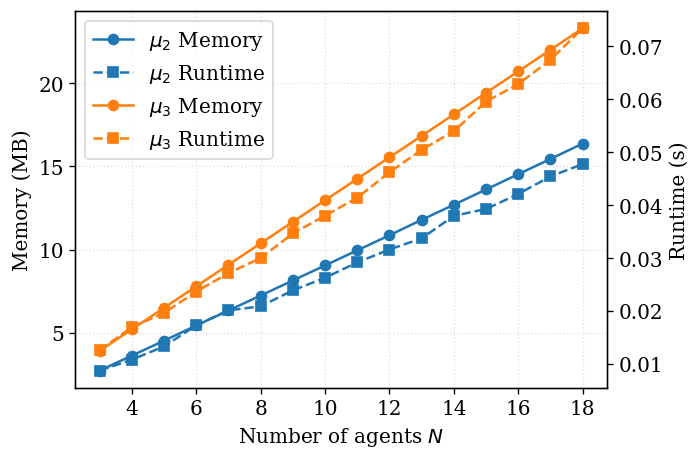}
    \caption{Scalability of the dual-tree approach for $\mu_2$ and $\mu_3$ with an increasing number of agents.}
    \label{fig:output_scalability34}
     \vspace{-0.5cm}
\end{figure}

 \begin{figure}[t]
    \centering
    \includegraphics[width=0.7\linewidth]{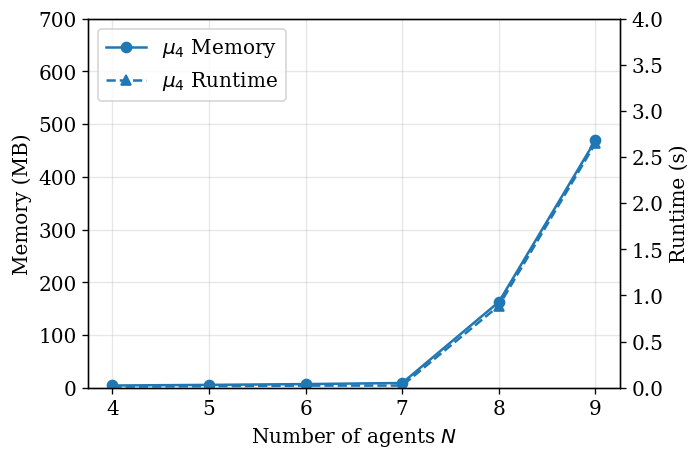}
    \caption{Scalability for $\mu_4$ with an increasing number of agents.}
    \label{fig:output_scalability5}
    \vspace{-0.5cm}
\end{figure}

\section{Conclusion and Future Work}
\label{sec:conclusion}
In this paper, we address the correct-by-design synthesis problem for decoupled homogeneous stochastic MAS under counting LTL specifications. We show that the satisfaction probability admits a structured tensor decomposition, and develop a dual-tree structure that exploits this decomposition for efficient value iteration. 
% We first showed that the satisfaction probability can be represented in the form of a tensor decomposition. Building on this representation, we developed a dual-tree structure for efficient value iteration, thereby enabling a compact and computationally efficient implementation of the synthesis procedure. 
Simulation results demonstrate that the proposed approach substantially reduces both the memory usage and computational overhead compared with existing methods, and scales 
effectively to large-scale scenarios. Future work will focus on addressing the complexity induced by DFA construction for cLTL specifications by investigating more efficient tree-pruning strategies to further enhance scalability.
\vspace{-0.1cm}

\bibliographystyle{IEEEtran}
\bibliography{references}
%\appendices

\appendix
\section{Proof}
\textbf{Proof of Lemma.~\ref{lem:sum}.}
\begin{proof}
\begin{align}\notag
    &\mbP^{x_0}_{\mdpM \times \mathrm{C}_{\boldsymbol{\pi} }}(\exists t \leq T: \word_{[0,t]}\in \setwitness )\\
    &= \mathbb{P}^{\mathrm{C}_{\boldsymbol{\pi}}}\!\left(\bigcup_{\word_{[0,t]}\in\setwitness, t\leq T}\{L(x[k]) = l[k],\, \forall t \! \in \! \llbracket 0,t \rrbracket |x[0] = x_0\}\right) \\\label{eq:mutualsum}
    &= \sum_{\word_{[0,t]}\in\setwitness, t\leq T}\!\mathbb{P}^{\mathrm{C}_{\boldsymbol{\pi} }}(L(x[k])\!=\! l[k],\forall k \! \in \! \llbracket 0,t \rrbracket|x[0]=x_0 )\\
    &=\sum_{\word_{[0,t]}\in\setwitness, t\leq T} P_{\word_{[0,t]}}^{\mathrm{C}_{\boldsymbol{\pi} }} (x_0)
\end{align}
where the equality in Eq.~\eqref{eq:mutualsum} holds since the set of witnesses
defines a union of disjoint events.
\end{proof}

\textbf{Proof of Theorem~\ref{the1}.}
\begin{proof}
% We prove the statement by induction on $k$.\\
% \emph{Base case.} 
For all $i\in \setagent$, $$\mathcal{G}_{\mathrm{m}0} = (\{1\}, \emptyset, q_f), \mathcal{G}_{\mathrm{s}0} = (\{1\}, \emptyset, \boldsymbol{1}), R(1,i) = 1. $$ It is then trivial that for all $i\in \setagent$, $$v_0^i(z_0)(x_0^i) = \mathbf{1} = W_0(R(z_0,i))(x_0^i).$$
Assume that the following holds 
\[
  v^i_k(z)(x_0^i) = W_k\big(R(z,i)\big)(x_0^i) = W_k(\kappa)(x_0^i), 
  \quad\forall (z,i)\in\mathcal{Z}\times\mathcal{I}.
\]    
We let $(z,l,\tilde z)$ be any edge created at this step, with
$q=L_Q(\tilde z)$ and $\tilde\kappa=R(\tilde z,i)$. Then for $k+1$, we have 
\begin{equation}
\begin{aligned}
v^i_{k+1}(\tilde z)(x^i) &= \op_{\letter^i}^{\pi_q^i}(v_k^i(z))(x^i)\\
&=\expectation^{x^{i'}}[\mathcal{L}_{\letter^i}(x^{i'})v_k^i(z)(x^{i'}) | x^i,a^i=\pi_q^i(x^i)] \\
&=\!  \mathbb{E}^{{x}^{i'}}[\mathcal{L}_{l^{i}}(x^{i'})W_k(\kappa)(x^{i'}) \!\mid \!x^{i},a^{i}\!=\!\pi_q^{i}(x^{i})]\\
&= \op^{\pi_q^{i}}_{l^{i}}(W_k\!(\kappa)\!)(x^{i}) \! = W_{k+1}(\tilde \kappa)(x^i)
\end{aligned}
\end{equation}
Hence, we conclude that for $k=0,1,\dots,T$ and all
$(z,i)\in\mathcal{Z} \times \mathcal{I}$, the following holds
\[
  v^i(z)(x_0^i) = W\big(R(z,i)\big)(x_0^i).
\]
By Lemma~\ref{lem:sum} and Proposition~\ref{prop:productP}, we have 
\begin{equation}
\begin{aligned}
\mathbb P^{x_0}\big(\mdpM \times \mathrm{C}^{\text{dec}}_{\bm{\pi}} \models \mu\big) &\geq \mbP^{x_0}_{\mdpM \times \mathrm{C}_{\boldsymbol{\pi} }}(\exists t \leq\! T\!:\! \word_{[0,t]}\in \setwitness) \\
&= \tensorV_{q}^T(x)\\
&=\sum_{z:\mathcal{L}_Q(z)=\bar{q}_0}  \prod_{ i\in \setagent}v^i(z)(x_0^i)\\
&=\sum_{z:\mathcal{L}_Q(z)=\bar{q}_0}  \prod_{ i\in \setagent} W(R(z,i))(x_0^i).
\end{aligned}
\end{equation}

\end{proof}

\end{document}